\documentclass[11pt]{llncs}
\usepackage{amsmath}
\usepackage{amssymb}
\usepackage{amsfonts}
\usepackage{eucal}
\usepackage{url}
\usepackage{amsfonts,amssymb}
\usepackage{pgf,tikz}
\usepackage{pgfplots}
\usepackage{cmap}
\usepackage{fancyhdr}
\usepackage{float}
\usetikzlibrary{arrows,automata,positioning,decorations.pathreplacing,decorations.markings}

\usepackage{tcolorbox}

\newcommand{\dt}{.}

\newtheorem{thm}{Theorem}

\newtheorem{lem}[thm]{Lemma}

\newcommand{\sa}{synchronizing automata}

\newcommand{\san}{synchronizing automaton}

\newcommand{\csw}{c.s.w}
\newcommand{\csws}{carefully synchronizing words}
\newcommand{\sws}{synchronizing words}

\newcommand{\ssws}{synchronizing words of minimum length}

\DeclareSymbolFont{rsfscript}{OMS}{rsfs}{m}{n}
\DeclareSymbolFontAlphabet{\mathrsfs}{rsfscript}

\newcommand{\mA}{\mathrsfs{A}}
\newcommand{\mB}{\mathrsfs{B}}

\begin{document}
\title{Using Sat solvers for synchronization issues\\ in partial deterministic automata\thanks{Supported by the Ministry of Science and Higher Education of the Russian Federation, projects no.\ 1.580.2016 and 1.3253.2017, and the Competitiveness  Enhancement Program of Ural Federal University.}}
\author{Hanan Shabana \and Mikhail V. Volkov}%

\titlerunning{Sat solvers for synchronization in partial automata}

\authorrunning{Hanan Shabana, M. V. Volkov}

\tocauthor{Hanan Shabana, M. V. Volkov (Ekaterinburg, Russia)}

\institute{Institute of Natural Sciences and Mathematics\\ Ural Federal University, Lenina 51, 620000 Ekaterinburg, Russia\\
\email{hananshabana22@gmail.com, m.v.volkov@urfu.ru}}

\maketitle

\begin{abstract}
We approach the task of computing a carefully synchronizing word of minimum length for a given partial deterministic automaton, encoding the problem as an instance of SAT and invoking a SAT solver. Our experimental results demonstrate that this approach gives satisfactory results for automata with up to 100 states even if very modest computational resources are used.

\keywords{Nondeterministic automaton, deterministic automaton, partial deterministic automaton, careful synchronization, exact synchronization, carefully synchronizing word, SAT, SAT-solver}
\end{abstract}

\section{Introduction}
\label{sec:intro}

A \emph{nondeterministic finite automaton} (NFA) is a triple $\langle Q,\Sigma,\delta\rangle$, where $Q$ and $\Sigma$ are finite non-empty sets called the \emph{state set} and the \emph{input alphabet} respectively, and $\delta$ is a subset of $Q\times\Sigma\times Q$. The elements of $Q$ and $\Sigma$ are called \emph{states} and \emph{letters}, respectively, and $\delta$ is referred to as the \emph{transition relation}\footnote{The conventional concept of an NFA includes distinguishing two non-empty subsets of $Q$ consisting of \emph{initial} and \emph{final} states. As these play no role in our considerations, the above simplified definition well suffices for the purpose of this paper.}. For each pair $(q,a)\in Q\times\Sigma$, we denote by $\delta(q,a)$ the subset $\{q'\mid (q,a,q')\in\delta\}$ of $Q$; this way $\delta$ can be viewed as a function $Q\times\Sigma\to\mathcal{P}(Q)$, where $\mathcal{P}(Q)$ is the power set of $Q$. When we treat $\delta$ as a function, we refer to it as the \emph{transition function}.

Let $\Sigma^*$ stand for the collection of all finite words over the alphabet $\Sigma$, including the empty word $\varepsilon$. The transition function extends to a function $\mathcal{P}(Q)\times\Sigma^*\to\mathcal{P}(Q)$, still denoted $\delta$, in the following inductive way: for every subset $S\subseteq Q$ and every word $w\in\Sigma^*$, we set
\[
\delta(S,w):=\begin{cases}S &\text{ if } w=\varepsilon, \\
\bigcup_{q\in\delta(S,v)}\delta(q,a)&\text{  if $w=va$ with $v\in\Sigma^*$ and $a\in\Sigma$.}
\end{cases}
\]
(Here the set $\delta(S,v)$ is defined by the induction assumption since $v$ is shorter than $w$.) We say that a word $w\in\Sigma^*$ is \emph{undefined at a state $q\in Q$} if the set $\delta(q,w)$ is empty; otherwise $w$ is said to be \emph{defined at $q$}.

When we deal with a fixed NFA, we suppress the sign of the transition relation, introducing the NFA as the pair $\langle Q,\Sigma\rangle$ rather than the triple $\langle Q,\Sigma,\delta\rangle$ and writing $q\dt w$ for $\delta(q,w)$ and $S\dt w$ for $\delta(S,w)$.

A \emph{partial} (respectively, \emph{complete}) \emph{ deterministic} automaton is an NFA $\langle Q,\Sigma\rangle$ such that $|q\dt a|\le 1$ (respectively, $|q\dt a|=1$) for all $(q,a)\in Q\times\Sigma$. We use the acronyms PFA and CFA for the expressions `partial deterministic automaton' and `complete deterministic automaton', respectively.

A CFA $\mA=\langle Q,\Sigma\rangle$ is called \emph{synchronizing} if there exists a word $w\in\Sigma^*$ whose action leaves the automaton in one particular state no matter at which state in $Q$ it is applied: $q\dt w=q'\dt w$ for all $q,q'\in Q$. Any $w$ with this property is said to be a \emph{synchronizing} word for the automaton.

Synchronizing automata serve as simple yet adequate models of error-resistant systems in many applied areas (system and protocol testing, information coding, robotics). At the same time, \sa\ surprisingly arise in some parts of pure mathematics and theoretical computer science (symbolic dynamics, theory of substitution systems, formal language theory). We refer to the survey~\cite{Vo08} and the chapter~\cite{KV} of the forthcoming `Handbook of Automata Theory' for a discussion of \sa\ as well as their diverse connections and applications. From both applied and theoretical viewpoints, the key question is to find the optimal, i.e., shortest reset word for a given \san. Under standard assumptions of complexity theory, this optimization question is known to be computationally hard; see \cite[Section~2]{KV} for a summary of various hardness results in the area. As it is quite common for hard problems of applied importance, there have been many attempts to develop practical approaches to the question. These approaches have been based on certain heuristics~\cite{Karahoda:2016,Altun:2017,Karahoda:2018} and/or popular techniques, including (but not limiting to) binary decision diagrams~\cite{Pixley&Jeong&Hachtel:1992}, genetic and evolutionary algorithms~\cite{Roman:2009,Kowalski:2017}, satisfiability solvers~\cite{Skvortsov:2011}, answer set programming~\cite{Gunicen:2013}, hierarchical classifiers~\cite{Podolak:2012}, and machine learning~\cite{Podolak:2018}.

The present authors~\cite{ShVo18,Sh18} have initiated an extension to the realm of NFAs of the approach of~\cite{Skvortsov:2011}. Here we consider a more restricted class, namely,  that of PFAs, where studying synchronization issues appears to be much better motivated. While we follow the general strategy of and re-use some technical tricks from~\cite{ShVo18,Sh18}, our present constructions heavily depend on the specifics of partial automata and have not been obtained via specializing the constructions of those earlier papers.

The rest of the paper is structured as follows. In Sect.~\ref{sec:synchro} we describe and motivate the version of PFA synchronization that we have studied. In Sect.~\ref{sec:encoding} we first outline the approach based on satisfiability solvers and then explain in detail how we encode PFAs and their synchronization problems as instances of the Boolean satisfiability problem. In Sect.~\ref{sec:experimental} we  provide samples of our experimental results and conclude in Sect.~\ref{sec:final} with a brief discussion of the future work.

We have tried to make the paper, to a reasonable extent, self-contained, except for a few discussions that involve some basic concepts of computational complexity theory. These concepts can be found, e.g., in the early chapters of the textbook~\cite{Papa}.

\section{Synchonization of NFAs and PFAs}
\label{sec:synchro}

The concept of synchronization of CFAs as defined in Sect.~\ref{sec:intro} was extended to NFAs in several non-equivalent ways. The following three nowadays popular versions were suggested and analyzed in~\cite{Imre99} in 1999 (although, in an implicit form, some of them appeared in the literature much earlier, see, e.g., \cite{Bu76,Goralcik82}). For $i\in\{1,2,3\}$, an NFA $\mathrsfs{A}=\langle Q,\Sigma\rangle$ is called $D_i$-\emph{syn\-chro\-nizing} if there exists a word $w\in\Sigma^*$ that satisfies the condition $(D_i)$ from the list below:
\begin{enumerate}
\itemindent=14pt
\item[$(D_{1})$:]  $|q.w|=|Q.w|=1$ for all $q\in Q$;
\item[$(D_{2})$:]  $q.w=Q.w$ for all $q\in Q$;
\item[$(D_{3})$:]  $\bigcap_{q\in Q}q.w\ne\varnothing$.
\end{enumerate}
Any word satisfying $(D_i)$ is called $D_{i}$-\emph{synchronizing} for $\mA$. The definition readily yields the following properties of $D_{i}$-synchronizing words:
\begin{lem}
\label{lem:di}
$a)$ A $D_1$- or $D_3$-synchronizing word is defined at each state.

$b)$ A $D_2$-synchronizing word is either defined at each state or undefined at each state.

$c)$ Every $D_1$-synchronizing word is both $D_2$- and $D_3$-synchronizing; every $D_2$-synchronizing word defined at each state is $D_3$-synchronizing.
\end{lem}

In~\cite{ShVo18} we adapted the approach based on satisfiability solvers to finding $D_3$-\ssws\ for NFAs. The first-named author used a similar method in the cases of
$D_1$- and $D_2$-synchronization; results related to $D_2$-synchronization were reported in~\cite{Sh18}.

Yet another version of synchronization for NFAs was introduced in~\cite{Ito04} and systematically studied in~\cite{Martyugin08,Martyugin10,Martyugin12,Martyugin13,Martyugin14}, which terminology we adopt. An NFA $\mathrsfs{A}=\langle Q,\Sigma\rangle$ is called \emph{carefully synchronizing} if there is a word $w=a_1\cdots a_\ell$, with $a_1,\dots,a_\ell\in\Sigma$, that satisfies the condition $(C)$, being the conjunction of $(C1)$--$(C3)$ below:
\begin{enumerate}
\itemindent=14pt
\item[$(C1)$:] the letter $a_1$ is defined at every state in $Q$;
\item[$(C2)$:] the letter $a_t$ with $1<t\le \ell$ is defined at every state in $Q.a_1\cdots a_{t-1}$,
\item[$(C3)$:] $|Q.w|=1$.
\end{enumerate}
Any $w$ satisfying $(C)$ is called a \emph{carefully synchronizing word} (\csw., for short) for $\mA$. Thus, when a \csw. is applied at any state in $Q$, no undefined transition occurs during the course of application. Every carefully synchronizing word is clearly $D_{1}$-synchronizing but the converse is not true in general; moreover, a $D_{1}$-synchronizing NFA may admit no \csw.

In this paper we focus on \csws\ for PFAs. There are several theoretical and practical reasons for this.

On the theoretical side, it is easy to see that each of the conditions $(C)$, $(D_{1})$, $(D_{3})$ leads to the same notion when restricted to PFAs. As for $D_2$-synchronization, if a word $w$ is $D_{2}$-synchronizing for a PFA $\mA$, then $w$ carefully synchronizes $\mA$ whenever $w$ is defined at each state. Otherwise $w$ is nowhere defined by Lemma~\ref{lem:di}b, and such `annihilating' words are nothing but usual \sws\ for the CFA obtained from $\mA$ by adding a new sink state and making all transitions undefined in $\mA$ lead to this sink state. Synchronization of CFAs with a sink state is relatively well understood (see~\cite{Ry97}), and therefore, we may conclude that $D_2$-synchronization also reduces to careful synchronization in the realm of PFAs. On the other hand, there exists a simple transformation that converts every NFA $\mathrsfs{A}=\langle Q,\Sigma\rangle$ into a PFA $\mathrsfs{B}=\langle Q,\Sigma'\rangle$ such that $\mA$ is $D_{3}$-synchronizing if and only if so is $\mB$ and the minimum lengths of $D_{3}$-synchronizing words for $\mA$ and $\mB$ are equal; see \cite[Lemma~8.3.8]{Ito} and \cite[Lemma~3]{Ito08}. These observations demonstrate that from the viewpoint of optimal synchronization, studying \csws\ for PFAs may provide both lower and upper bounds applicable to arbitrary NFAs and all aforementioned kinds of synchronization.

Probably even more important is the fact that careful synchronization of PFAs is relevant in numerous applications. Due to the page limit, we mention only two examples here.

In industrial robotics, \sa\ are widely used to design feeders, sorters, and orienters that work with flows of certain objects carried by a conveyer. The goal is achieved by making the flow encounter passive obstacles placed appropriately along the conveyer belt; see \cite{Na86,Na89} for the origin of this automata approach and \cite{AnVo04} for an illustrative example. Now imagine that the objects to be oriented or sorted have a fragile side that could be damaged if hitting an obstacle. In order to prevent any damage, we have to forbid `dangerous' transitions in the automaton modelling the orienter/sorter so that the automaton becomes partial and the obstacle sequences must correspond to \csws. (Actually, the term `careful synchronization' has been selected with this application in mind.)

Our second example relates to so-called synchronized codes\footnote{We refer the reader to \cite[Chapters~3 and~10]{Berstel&Perrin&Reutenauer:2009} for a detailed account of profound connections between codes and automata.}. Recall that a \emph{prefix code} over a finite alphabet $\Sigma$ is a set $X$ of words in $\Sigma^*$ such that no word of $X$ is a prefix of another word of $X$. Decoding of a finite prefix code $X$ over $\Sigma$ can be implemented by a finite deterministic automaton $\mA_X$ whose state $Q$ is the set of all proper prefixes of the words in $X$ (including the empty word $\varepsilon$) and whose transitions are defined as follows: for $q\in Q$ and $a\in\Sigma$,
\begin{displaymath}
q\dt a =\begin{cases} qa & \text{if $qa$ is a proper prefix of a word of $X$},\\
\varepsilon & \text{if $qa \in X$},\\
\text{undefined} & \text{otherwise}.
\end{cases}
\end{displaymath}
In general, $\mA_X$  is a PFA (it is complete if and only if the code $X$ is not contained in another prefix code over $\Sigma$). It can be shown that if $\mA_X$ is carefully synchronizing, the code $X$ has a very useful property: whenever a loss of synchronization between the decoder and the coder occurs (because of a channel error), it suffices to transmit a \csw. $w$ of $\mA_X$ such that $w$ sends all states in $Q$ to the state $\varepsilon$ to ensure that the following symbols will be decoded correctly.

We may conclude that the problems of determining whether or not a given PFA is carefully synchronizing and of finding its shortest \csws\ are both natural and important. The bad news is that these problems turn out to be quite difficult: it is known that the first problem is PSPACE-complete and that the minimum length of \csws\ for carefully synchronizing PFAs can be exponential as a function of the number of states. (These results were found in~\cite{Rystsov:1980,Rystsov:1983} and later rediscovered and strengthened in~\cite{Martyugin12}.) Thus, one has to use  some tools that have proved to be efficient for dealing with computationally hard problems.  As mentioned in Section~\ref{sec:intro}, in this paper we make an attempt to employ a satisfiability solver as such a tool.

\section{Encoding}
\label{sec:encoding}

For completeness, recall the formulation of the Boolean satisfiability problem (SAT). An instance of SAT is a pair $(V,C)$, where $V$ is a set of Boolean variables and $C$ is a collection of clauses over $V$. (A \emph{clause} over $V$ is a disjunction of literals and a \emph{literal} is either a variable in $V$ or the negation of a variable in~$V$.) Any \emph{truth assignment} on $V$, i.e., any map $\varphi\colon V\to\{0,1\}$, extends to a map $C\to\{0,1\}$ (still denoted by $\varphi$) via the usual rules of propositional calculus: $\varphi(\neg x)=1-\varphi(x)$, $\varphi(x\vee y)=\max\{\varphi(x),\varphi(y)\}$. A truth assignment $\varphi$ \emph{satisfies} $C$ if $\varphi(c)=1$ for all $c\in C$. The answer to an instance $(V,C)$ is YES if $(V,C)$ has a \emph{satisfying assignment} (i.e., a truth assignment on $V$ that satisfies $C$) and NO otherwise.

By Cook's classic theorem (see, e.g., \cite[Theorem 8.2]{Papa}), SAT is NP-com\-plete, and by the very definition of NP-completeness, every problem in NP reduces to SAT. On the other hand, over the last score or so, many efficient \emph{SAT-solvers}, i.e., specialized programs designed to solve instances of SAT have been developed. Modern SAT solvers can solve instances with hundreds of thousands of variables and millions of clauses within a few minutes. Due to this progress, the following approach to computationally hard problems has become quite popular nowadays: one encodes instances of such problems into instances of SAT that are then fed to a SAT-solver\footnote{We refer the reader to the survey~\cite{HKR} a detailed discussion of the approach and examples of its successful applications in various areas.}. It is exactly the strategy that we want to apply.

We start with the following problem:

\smallskip

\noindent CSW (the existence of a \csw. of a given length):

\noindent\textsc{Input}: a PFA $\mathrsfs{A}$ and a positive integer $\ell$ (given in unary);

\noindent\textsc{Output}: YES if $\mathrsfs{A}$ has a \csw. of length $\ell$;\\\phantom{\textsc{Output}:} NO otherwise.

\smallskip

\noindent We have to assume that the integer $\ell$ is given in unary because with $\ell$ given in binary, a polynomial time reduction from CSW to SAT is hardly possible. (Indeed, it easily follows from~\cite{Martyugin12} that the version of CSW in which the integer parameter is given in binary is PSPACE-hard, and the existence of a polynomial reduction from a PSPACE-hard problem to SAT would imply that the polynomial hierarchy collapses at level~1.) In contrast, the version of CSW with the unary integer parameter is easily seen to belong to NP: given an instance $(\mathrsfs{A}=\langle Q,\Sigma\rangle,\ell)$ of CSW in this setting, guessing a word $w\in\Sigma^*$ of length $\ell$ is legitimate. Then one just checks whether or not $w$ is carefully synchronizing for $\mathrsfs{A}$, and time spent for this check is clearly polynomial in the size of $(\mathrsfs{A},\ell)$.

Now, given an arbitrary instance $(\mathrsfs{A},\ell)$ of CSW, we construct an instance $(V,C)$ of SAT such that the answer to $(\mathrsfs{A},\ell)$ is YES if and only if so is the answer to $(V,C)$. In the following presentation of our encoding, precise definitions and statements are interwoven with less formal comments explaining the `physical' meaning of variables and clauses.

So, take a PFA $\mathrsfs{A}=\langle Q,\Sigma\rangle$ and an integer $\ell>0$. Denote the sizes of $Q$ and $\Sigma$ by $n$ and $m$ respectively, and fix some numbering of these sets so that $Q=\{q_1,\dots,q_n\}$ and $\Sigma=\{a_1,\dots,a_m\}$.

We start with introducing the variables used in the instance $(V,C)$ of SAT that encodes $(\mathrsfs{A},\ell)$. The set $V$ consists of two sorts of variables: $m\ell$ \emph{letter variables} $x_{i,t}$ with $1\le i\le m$, $1\le t\le\ell$, and $n(\ell+1)$ \emph{state variables} $y_{j,t}$ with $1\le j\le n$, $0\le t\le\ell$. We use the letter variables to encode the letters of a hypothetical \csw. $w$ of length $\ell$: namely, we want the value of the variable $x_{i,t}$ to be 1 if and only if the $t$-th letter of $w$ is $a_i$. The intended meaning of the state variables is as follows: we want the value of the variable $y_{j,t}$ to be 1 whenever the state $q_j$ belongs to the image of $Q$ under the action of the prefix of $w$ of length $t$, in which situation we say that $q_j$ \emph{is active after $t$ steps}. We see that the total number of variables in $V$ is $m\ell+n(\ell+1)=(m+n)\ell+n$.

Now we turn to constructing the set of clauses $C$. It consists of four groups. The group $I$ of \emph{initial clauses} contains $n$ one-literal clauses $y_{j,0}$, $1\le j\le n$, and expresses the fact that all states are active after 0 steps.

For each $t=1,\dots,\ell$, the group $L$ of \emph{letter clauses} includes the clauses
\begin{equation}
\label{eq:only one letter}
x_{1,t}\vee\dots\vee x_{m,t},\quad \neg x_{r,t}\vee\neg x_{s,t},\ \text{ where }\ 1\le r<s\le m.
\end{equation}
Clearly, the clauses \eqref{eq:only one letter} express the fact that the $t$-th position of our hypothetical \csw. $w$ is occupied by exactly one letter in $\Sigma$. Altogether, $L$ contains $\ell\left(\frac{m(m-1)}2+1\right)$ clauses.

For each $t=1,\dots,\ell$ and each triple $(q_j,a_i,q_k)$ in the transition relation of $\mA$, the group $T$ of \emph{transition clauses} includes the clause
\begin{equation}
\label{eq:transition}
\neg y_{j,t-1}\vee\neg x_{i,t}\vee y_{k,t}.
\end{equation}
Invoking the basic laws of propositional logic, one sees that the clause \eqref{eq:transition} is equivalent to the implication
$y_{j,t-1}\mathop{\&} x_{i,t}\to y_{k,t},$
that is, \eqref{eq:transition} expresses the fact that if the state $q_j$ has been active after $t-1$ steps and $a_i$ is the $t$-th letter of $w$, then the state $q_k=q_j\dt a_i$ becomes active after $t$ steps. Further, for each $t=1,\dots,\ell$ and each pair $(q_j,a_i)$ such that $a_i$ is undefined at $q_j$ in $\mA$, we add to $T$ the clause
\begin{equation}
\label{eq:undefined}
\neg y_{j,t-1}\vee\neg x_{i,t}.
\end{equation}
The clause is equivalent to the implication
$y_{j,t-1}\to\neg x_{i,t},$
and thus, it expresses the requirement that the letter $a_i$ should not be occur in the $t$-th position of $w$ if $q_j$ has been active after $t-1$ steps. Obviously, this corresponds to the conditions $(C1)$ (for $t=0$) and $(C2)$ (for $t>0$) in the definition of careful synchronization. For each $t=1,\dots,\ell$ and each pair $(q_j,a_i)\in Q\times\Sigma$, exactly one of the clauses \eqref{eq:transition} or \eqref{eq:undefined} occurs in $T$, whence $T$ consists of $\ell mn$ clauses.

The final group $S$ of of \emph{synchronization clauses} includes the clauses
\begin{equation}
\label{eq:only one state}
\neg y_{r,\ell}\vee\neg y_{s,\ell},\ \text{ where }\ 1\le r<s\le n.
\end{equation}
The clauses \eqref{eq:only one state} express the requirement that at most one state remains active when the action of the word $w$ is completed, which corresponds to the condition $(C3)$ from the definition of careful synchronization. The group $S$ contains $\frac{n(n-1)}2$ clauses.

Summing up, the number of clauses in $C:=I\cup L\cup T\cup S$ is
\begin{multline}
\label{eq:number}
n+\ell\left(\tfrac{m(m-1)}2+1\right)+\ell mn+\tfrac{n(n-1)}2=\\
\ell\left(\tfrac{m(m-1)}2+mn+1\right)+\tfrac{n(n+1)}2.
\end{multline}
In comparison with encodings used in our earlier papers~\cite{ShVo18,Sh18}, the encoding suggested here produces much smaller SAT instances. Since in the applications the size
of the input alphabet is a (usually small) constant, the leading term in \eqref{eq:number} is $\Theta(\ell n)$ while the restriction to PFAs of the encodings from~\cite{ShVo18,Sh18} has $\Theta(\ell n^2)$ clauses.

\begin{thm}
\label{thm:reduction}
A PFA $\mathrsfs{A}$ has a \csw. of length $\ell$ if and only if the instance $(V,C)$ of SAT constructed above is satisfiable. Moreover, the \csws\ of length $\ell$ for $\mathrsfs{A}$ are in a 1-1 correspondence with the restrictions of satisfying assignments of $(V,C)$ to the letter variables.
\end{thm}

\begin{proof}
Suppose that $\mathrsfs{A}$ has a \csw. of length $\ell$. We fix such a word $w$ and denote by $w_t$ its prefix of length $t=1,\dots,\ell$. Define a truth assignment $\varphi\colon V\to\{0,1\}$ as follows: for $1\le i\le m$, $0\le j\le n$, $1\le t\le\ell$, let
\begin{align}
\label{eq:letter variables}
\varphi(x_{i,t})&:=\begin{cases}
1 &\text{if the $t$-th letter of $w$ is $a_i$,}\\
0 &\text{otherwise;}
\end{cases}\\
\label{eq:initial variables}
\varphi(y_{j,0})&:=1;\\
\label{eq:state variables}
\varphi(y_{j,t})&:=\begin{cases}
1 &\text{if the state $q_j$ lies in $Q\dt w_t$,}\\
0 &\text{otherwise.}
\end{cases}
\end{align}
In view of \eqref{eq:letter variables} and \eqref{eq:initial variables}, $\varphi$ satisfies all clauses in $L$ and respectively $I$. As $w_\ell=w$ and $|Q\dt w|=1$, we see that \eqref{eq:state variables} ensures that $\varphi$ satisfies all clauses in $S$. It remains to analyze the clauses in $T$. For each fixed $t=1,\dots,\ell$, these clauses are in a 1-1 correspondence with the pairs in $Q\times\Sigma$. We fix such a pair $(q_j,a_i)$, denote the clause corresponding to  $(q_j,a_i)$ by $c$ and consider three cases.

\emph{\textbf{Case 1}: the letter $a_i$ is not the $t$-th letter of $w$}. In this case $\varphi(x_{i,t})=0$ by \eqref{eq:letter variables}, and hence, $\varphi(c)=1$ since the literal $\neg x_{i,t}$ occurs in $c$, independently of $c$ having the form \eqref{eq:transition} or \eqref{eq:undefined}.

\emph{\textbf{Case 2}: the letter $a_i$ is the $t$-th letter of $w$ but it is undefined at $q_j$}. In this case the clause $c$ must be of the form \eqref{eq:undefined}. Observe that $t>1$ in this case since the first letter of the \csw. $w$ must be defined at each state in $Q$. Moreover, the state $q_j$ cannot belong to the set $Q\dt w_{t-1}$ because $a_i$ must be defined at each state in this state. Hence $\varphi(y_{j,t-1})=0$ by \eqref{eq:state variables}, and $\varphi(c)=1$ since the literal $\neg y_{j,t-1}$ occurs in $c$.

\emph{\textbf{Case 3}: the letter $a_i$ is the $t$-th letter of $w$ and it is defined at $q_j$}. In this case the clause $c$ must be of the form \eqref{eq:transition}, in which the literal $y_{k,t}$ corresponds to the state $q_k=q_j\dt a_i$. If the state $q_j$ does not belong to the set $Q\dt w_{t-1}$, then as in the previous case, we have $\varphi(y_{j,t-1})=0$ and $\varphi(c)=1$. If $q_j$ belongs to $Q\dt w_{t-1}$, then the state $q_k$ belongs to the set $(Q\dt w_{t-1})\dt a_i=Q\dt w_t$, whence $\varphi(y_{k,t})=1$ by \eqref{eq:state variables}. We conclude that $\varphi(c)=1$ since the literal $y_{k,t}$ occurs in $c$.

\smallskip

Conversely, suppose that $\varphi\colon V\to\{0,1\}$ is a satisfying assignment for $(V,C)$. Since $\varphi$ satisfies the clauses in $L$, for each $t=1,\dots,\ell$, there exists a unique $i\in\{1,\dots,m\}$ such that $\varphi(x_{i,t})=1$. This defines a map $\chi\colon\{1,\dots,\ell\}\to\{1,\dots,m\}$. Let $w:=a_{\chi(1)}\cdots a_{\chi(\ell)}$. We aim to show that $w$ is a \csw. for $\mA$, i.e., to verify that $w$ fulfils the conditions $(C1)$--$(C3)$ from the definition of a \csw. For this, we first prove two auxiliary claims. Recall that a state is said to be active after $t$ steps if it lies in $Q\dt w_t$, where, as above, $w_t$ is the length $t$ prefix of the word $w$. (By the length 0 prefix we understand the empty word $\varepsilon$.)

\smallskip

\emph{\textbf{Claim 1}. For each $t=0,1,\dots,\ell$, there are states active after $t$ steps}.

\emph{\textbf{Claim 2}. If a state $q_k$ is active after $t$ steps, then $\varphi(y_{k,t})=1$}.

We prove both claims simultaneously by induction on $t$. The induction basis $t=0$ is guaranteed by the fact that all states are active after 0 steps and $\varphi$ satisfies the clauses in $I$. Now suppose that $t>0$ and there are states active after $t-1$ steps. Let $q_r$ be such a state. Then $\varphi(y_{r,t-1})=1$ by the induction assumption. Let $i:=\chi(t)$, that is, $a_i$ is the $t$-th letter of the word $w$. Then $\varphi(x_{i,t})=1$, whence $\varphi$ cannot satisfy the clause of the form \eqref{eq:undefined} with $j=r$. Hence this clause cannot appear in $T$ as $\varphi$ satisfies the clauses in $T$. This means that the letter $a_i$ is defined at $q_r$ in $\mA$, and the state $q_s:=q_r\dt a_i$ is active after $t$ steps. Claim 1 is proved.

Now let $q_k$ be an arbitrary state that is active after $t>0$ steps. Since $a_i$ is the $t$-th letter of $w$, we have $Q\dt w_t=(Q\dt w_{t-1})\dt a_i$, whence
$q_k=q_j\dt a_i$ for same $q_j\in Q\dt w_{t-1}$. Therefore the clause \eqref{eq:transition} occurs in $T$, and thus, it is satisfied by $\varphi$. Since $q_j$ is active after $t-1$ steps, $\varphi(y_{j,t-1})=1$ by the induction assumption; besides that, $\varphi(x_{i,t})=1$. We conclude that in order to satisfy \eqref{eq:transition}, the assignment $\varphi$ must fulfil $\varphi(y_{k,t})=1$. This completes the proof of Claim~2.

\smallskip

We turn to prove that the word $w$ fulfils $(C1)$ and $(C2)$. This amounts to verifying that for each $t=1,\dots,\ell$, the $t$-th letter of the word $w$ is defined at every state $q_j$ that is active after $t-1$ steps. Let, as above, $a_i$ stand for the $t$-th letter of $w$. If $a_j$ were undefined at $q_j$, then by the definition of the set $T$ of transition clauses, this set would include the corresponding clause \eqref{eq:undefined}. However, $\varphi(x_{i,t})=1$ by the construction of $w$ and $\varphi(y_{j,t-1})=1$ by Claim~2. Hence $\varphi$ does not satisfy this clause while the clauses from $T$ are satisfied by $\varphi$, a contradiction.

Finally, consider $(C3)$. By Claim~1, some state is active after $\ell$ steps. On the other hand, the assignment $\varphi$ satisfies the clauses in $S$, which means that $\varphi(y_{j,\ell})=1$ for at most one index $j\in\{1,\dots,n\}$. By Claim~2 this implies that at most one state is active after $\ell$ steps. We conclude that exactly one state is active after $\ell$ steps, that is, $|Q\dt w|=1$.
\qed
\end{proof}

\section{Experimental results}
\label{sec:experimental}
We have successfully applied the encoding constructed in Sect.~\ref{sec:encoding} to solve CSW instances with the help of a SAT-solver. As in~\cite{Skvortsov:2011,Gunicen:2013,ShVo18,Sh18}, we have used MiniSat 2.2.0~\cite{Minisat,Minisat_page}. In order to find a \csw. of minimum length for a given PFA $\mA$, we have considered CSW instances $(\mA,\ell)$ with fixed $\mA$ and performed binary search on $\ell$.  Even though our encoding is different form those we used in \cite{ShVo18,Sh18}, it shares with them the following useful feature: when presented in DIMACS CNF format, the `primary' SAT instance that encodes the CSW instance $(\mA,1)$ can be easily scaled to the SAT instances that encode the CSW instances $(\mA,\ell)$ with any value of~$\ell$. Due to this feature, one radically reduces time needed to prepare the input data for the SAT-solver. We refer the reader to \cite[Sect.~3]{Sh18} for a detailed explanation of the trick and an illustrative example.

We implemented the algorithm outlined above in C++ and compiled with GCC 4.9.2. In our experiments we used a personal computer with an Intel(R) Core(TM) i5-2520M processor with 2.5 GHz CPU and 4GB of RAM. The code can be found at \url{https://github.com/hananshabana/SynchronizationChecker}.

As a sample of our experimental findings, we present here our results on synchronization of PFAs with a unique undefined transition. Observe that the problem of deciding whether or not a given PFA is carefully synchronizing remains PSPACE-complete even if restricted to this rather special case \cite{Martyugin12}. We considered random PFAs with $n\le 100$ states and two input letters. The condition $(C1)$ in the definition of a carefully synchronizing PFA implies that such a PFA must have an everywhere defined letter. We denoted this letter by $a$ and the other letter, called $b$, was chosen to be undefined at a unique state. Further, it is easy to see that for a PFA $\langle Q,\{a,b\}\rangle$ with $a,b$ so chosen to be carefully synchronizing, it is necessary that $|Q\dt a|<|Q|$. Therefore, we fixed a state $q_a\in Q$ and then selected $a$ uniformly at random from all $n^{n-1}$ maps $Q\to Q\setminus\{q_a\}$. Similarly, to ensure there is a unique undefined transition with $b$, we fixed a state $q_b\in Q$ (not necessarily different from $q_a$) and then selected $b$ uniformly at random from all $(n-1)^n$ maps $Q\setminus\{q_b\}\to Q$. For each fixed $n$, we generated up to 1000 random PFAs this way and calculated the average length $\ell(n)$ of their shortest \csws. We used the least squares method to find a function that best reflects how $\ell(n)$ depends on $n$, and it turned out that our results are reasonably well approximated by the following expression:
\[
\ell(n)\approx 3.92+0.49n-0.005n^2+0.000024n^3.
\]

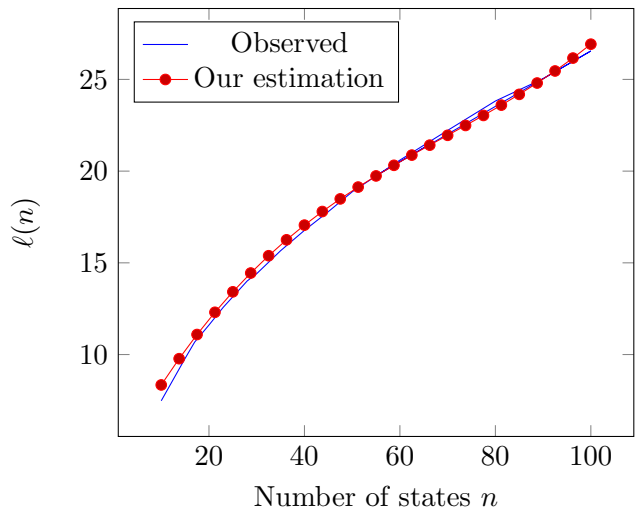
\begin{tikzpicture}
\begin{axis}[
xlabel= {Number of states $n$},
ylabel={$\ell(n)$},
domain=10:100,
legend pos=north west
]
\addplot+[mark=.]
plot coordinates {
	(10,7.480)
	(15,9.790)
	(17,10.680)
	(20,11.610)
	(23,12.580)
	(25,13.150)
	(28,14.010)
	(30,14.410)
	(35,15.660)
	(40,16.770)
	(45,17.790)
	(50,18.870)
	(60,20.600)
	(70,22.220)
	(80,23.820)
	(90,25.040)
	(100,26.550)
	(55,19.750)
	(65,21.280)
		
};
\addlegendentry{Observed}

\addplot +[mark=*][domain=10:100] {3.92+(0.49*x)+(-0.005*x^2)+(0.000024*x^3)};

\addlegendentry{Our estimation}
\end{axis}
\end{tikzpicture}

The next graph shows the relation between the relative standard deviation of our datasets and the number of states. We see that the relative standard deviation gradually decreases as the number of states grows.

\medskip

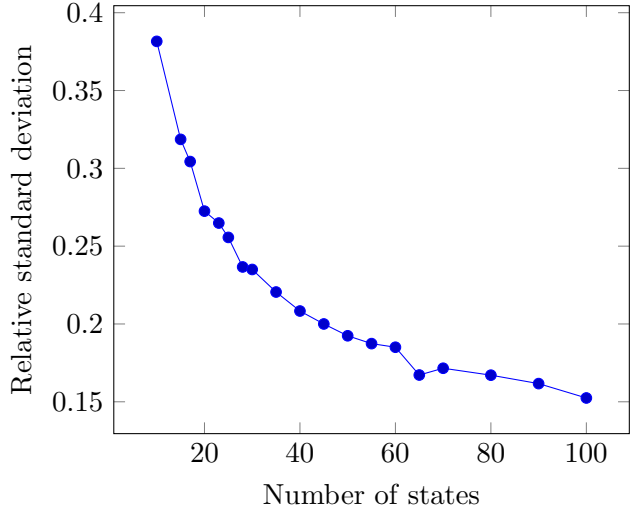
\begin{tikzpicture}
\begin{axis}[
xlabel= {Number of states },
ylabel={Relative standard deviation},
legend pos=outer north east
]
\addplot+[mark=*]
plot coordinates {
	(10,0.38159225)
	(15,0.318618803)
	(17,0.304359789)
	(20,0.272500964)
	(23,0.264849967)
	(25,0.255621131)
	(28,0.236635543)
	(30,0.235014587)
	(35,0.220527863)
	(40,0.208323052)
	(45,0.200001709)
	(50,0.192410953)
	(55,0.187375028)
	(60,0.185061504)
	(65,0.167199559)
	(70,0.171558502)
	(80,0.167102896)
	(90,0.161681628)
	(100,0.152451951)
};
\end{axis}
\end{tikzpicture}

We performed similar experiments with random PFAs that have two or three undefined transition. We also tested our algorithm on PFAs from the series $\mathrsfs{P}_n$ suggested in~\cite{dBDZ1}. The state set of $\mathrsfs{P}_n$ is $\{1,2,\dots,n\}$, $n\ge 3$, on which the input letters $a$ and $b$ act as follows:
\[
q.a:=
     \begin{cases}
   q+1  &\text{if }  q=1,2, \\
     q  &\text{if }  q=3,\dots,n;
    \end{cases}\quad
q.b:=
    \begin{cases}
 \text{undefined} &\text{if } q=1, \\
  q+1  & \text{if } q=2,\dots,n-1,\\
  1 &\text{if } q=n.
    \end{cases}
\]
We examined all automata $\mathrsfs{P}_n$ with $n=4,5,\dots,11$, and for each of them, our result matched the theoretical value predicted by~\cite[Theorem~3]{dBDZ1}. The time consumed ranges from 0.301 sec for $n=4$ to 4303 sec for $n=11$. Observe that in the latter case the shortest \csw. has length 116 so that honest binary search started with $(\mathrsfs{P}_{11},1)$ required 11 iterations.

We made also a comparison with the only approach to computing \csws\ of minimum length that we had found in the literature, namely, the approach based on partial power automata. Given a PFA $\mA=\langle Q,\Sigma\rangle$, its \emph{partial power automaton} $\mathcal{P}(\mA)$ has the subsets of $Q$ as the states, the same input alphabet $\Sigma$, and the transition function defined as follows: for each $a\in\Sigma$ and each $P\subseteq Q$,
\[
P{\dt}a:=\begin{cases}
    \{q{\dt}a\mid q\in P\} & \text{provided $q{\dt}a$ is defined for all $q\in P$}, \\
    \text{undefined} & \text{otherwise}.
  \end{cases}
\]
It is easy to see that  $w\in\Sigma^*$ is a \csw. of minimum length for $\mA$ if and only if $w$ labels a minimum length path in $\mathcal{P}(\mA)$ starting at $Q$ and ending at a singleton. Such a path can be found by breadth-first search in the underlying digraph of $\mathcal{P}(\mA)$.

The result of the comparison are presented in the picture below. In this experiment we had to restrict to PFAs with at most 16 states since beyond this number of states, our implementation of the method based on partial power automata could not complete the computation due to memory restrictions (recall that we used rather modest computational resources). However, we think that the exhibited data suffice to demonstrate that the approach based on SAT-solvers shows a by far better performance.

\medskip

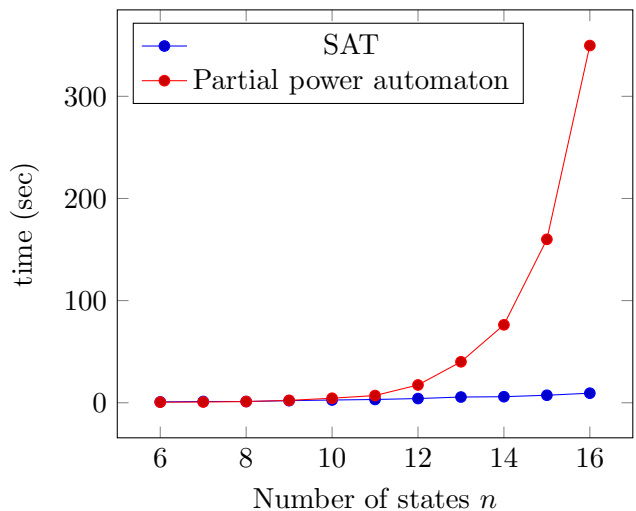
\begin{tikzpicture}
\begin{axis}[
xlabel= {Number of states $n$},
ylabel={time (sec)},
legend pos=north west
]
\addplot+[mark=*]
plot coordinates {
	(6,0.82)
	(7,1.26)
	(8,1.2)
	(9,2.059)
	(10,2.68)
	(11,3.2)
	(12,4.12)
	(13,5.6)
	(14,5.93)
	(15,7.36)
	(16,9.37)
	
};
\addlegendentry{SAT }

\addplot+[mark=*]
plot coordinates {
	(6,0.59)
	(7,0.72)
	(8,1.2)
	(9,2.32)
	(10,4.45)
	(11,7.06)
	(12,17.41)
	(13,40.08)
	(14,76.3)
	(15,160)
	(16,349.7)

};
\addlegendentry{Partial power automaton}
\end{axis}
\end{tikzpicture}

\section{Conclusion and future work}
\label{sec:final}

We have presented an attempt to approach the problem of computing a \csw. of minimum length for a given PFA via the SAT-solver method. For this, we have developed a new encoding, which, in comparison with encodings used in our earlier papers~\cite{ShVo18,Sh18}, requires a more sophisticated proof but leads to more economic SAT instances.

We plan to continue our experiments. In particular, it is interesting to compare the minimum lengths of a synchronizing word for a synchronizing DFA and of \csws\ for PFAs that can be obtained from the DFA by removing one or more of its transitions.

We also plan to extend the  SAT-solver approach to so-called \emph{exact synchronization} of PFAs which is of interest for certain applications.

\end{document}